\documentclass[12pt]{article}
\usepackage{amsmath, amsthm}
\usepackage{graphicx,psfrag,epsf}
\usepackage{enumerate}

\newcommand{\blind}{0}

\newtheorem{theorem}{\bf Theorem}

\newtheorem{lemma}{\bf Lemma}

\numberwithin{equation}{section}

\begin{document}

\def\spacingset#1{\renewcommand{\baselinestretch}%
{#1}\small\normalsize} \spacingset{1}


\if0\blind
{\title{On a Multiplicative Algorithm for Computing Bayesian D-optimal Designs}
  \author{Yaming Yu\\
    Department of Statistics, University of California, Irvine}
  \maketitle
} \fi

\if1\blind
{
  \bigskip
  \bigskip
  \bigskip
  \begin{center}
    {\LARGE\bf On a Multiplicative Algorithm for Computing Bayesian D-optimal Designs}
\end{center}
  \medskip
} \fi

\bigskip
\begin{abstract}
We use the minorization-maximization principle (Lange, Hunter and Yang 2000) to establish the monotonicity of a multiplicative algorithm for computing Bayesian D-optimal designs.  This proves a conjecture of Dette, Pepelyshev and Zhigljavsky (2008). 
\end{abstract}

\noindent%
\spacingset{1.45}

\section{Introduction}
Multiplicative algorithms (Silvey, Titterington and Torsney 1978; Torsney and Mandal 2006; Dette, Pepelyshev and Zhigljavsky 2008) are often employed in numerical computation of optimal designs (approximate theory; see Kiefer 1974, Silvey 1980, and Pukelsheim 1993).  These iterative algorithms are simple, easy to implement, and often increase the optimality criterion monotonically.  In the case of D-optimality, for example, monotonicity of the algorithm of Silvey et al.\ (1978) is well known (Titterington 1976; P\'{a}zman 1986); see Yu (2010a) and the references therein for further results.  Monotonicity is an important property as it implies convergence under mild conditions.

Bayesian D-optimality is a widely used design criterion that can accommodate prior uncertainty in the parameters (see Chaloner and Larntz 1989 and Chaloner and Verdinelli 1995).  Multiplicative algorithms extend naturally from D-optimality to Bayesian D-optimality (Dette et al.\ 2008).  Although the form of the algorithms is just as simple as in the D-optimal case, a corresponding monotonicity result is still lacking.  In the context of nonlinear regression, Dette et al.\ (2008) conjecture the monotonicity of a class of algorithms for computing Bayesian D-optimal designs.  The main theoretical contribution of this work is to confirm their monotonicity conjecture. 

Our technical devices include convexity and the minorization-maximization principle (MM; Lange, Hunter and Yang 2000; Hunter and Lange 2004).  Similar ideas play a key role in settling the related Titterington's (1978) conjecture (see Yu 2010a, 2010b).  Minorization-maximization (or bound optimization) is a general method for constructing iterative algorithms that increase an objective function $\phi(w)$ monotonically.  We first construct a function $Q(w; \tilde{w})$ such that $\phi(w)\geq Q(w; \tilde{w})$ for all $w$ and $\tilde{w}$, and $\phi(w)=Q(w; w)$.  Suppose the current iterate is $w^{(t)}$.  We choose $w^{(t+1)}$ to increase the $Q$ function, i.e., 
\begin{equation}
\label{qincrease}
Q\left( w^{(t+1)}; w^{(t)}\right)\geq Q\left( w^{(t)}; w^{(t)}\right).
\end{equation}
Then $w^{(t+1)}$ also increases the objective function $\phi$, because 
$$\phi\left(w^{(t+1)}\right)\geq Q\left( w^{(t+1)}; w^{(t)}\right)\geq Q\left( w^{(t)}; w^{(t)}\right) =\phi\left(w^{(t)}\right).$$
The usual MM algorithm chooses $w^{(t+1)}$ to maximize $Q\left(\cdot; w^{(t)}\right)$.  Since we only require (\ref{qincrease}), it is proper to call this strategy a {\it general MM algorithm}.  The general MM algorithm is an extension of the general expectation-maximization algorithm (GEM; Dempster, Laird and Rubin 1977). 

In Section~2 we state our monotonicity result and illustrate with a simple logistic regression example.  Section~3 proves the monotonicity result.  Specifically, the algorithm of Dette et al.\ (2008) for computing Bayesian D-optimal designs is derived as a general MM algorithm. 

\section{Theoretical Result and Illustration}
We focus on a finite design space $\mathcal{X}=\{x_1, \ldots, x_n\}.$  Let $\theta$ be the $m\times 1$ parameter of interest, and let $A_i(\theta)$ denote the $m\times m$ Fisher information matrix provided by a unit assigned to design point $x_i$.  The so-called Bayesian D-optimality (Chaloner and Larntz 1989) seeks to maximize 
\begin{equation*}
\phi (w)\equiv \int \log\det M(w, \theta)\, {\rm d}\pi(\theta),
\end{equation*}
where $\pi(\theta)$ is a probability distribution representing prior knowledge about $\theta$, and 
$$M(w, \theta)=\sum_{i=1}^n w_i A_i(\theta).$$
This is an extension of local D-optimality which chooses the design weights $w_i$ to maximize the log-determinant of the Fisher information for a fixed $\theta$.  It can also be viewed as a large sample approximation to Lindley's (1956) criterion based on Shannon information.  Here $w=(w_1,\ldots, w_n)\in \bar{\Omega},$ and $\bar{\Omega}$ denotes the closure of $\Omega=\{w:\ \sum_{i=1}^n w_i=1,\ w_i> 0\}$.  To convert $w$ to a finite-sample design, some rounding procedure is needed (Pukelsheim 1993, Chapter 12).  The matrices $A_i(\theta)$ are assumed to be well defined and nonnegative definite for every $\theta$. 

Let us consider the following algorithm for maximizing $\phi(w)$.  Define 
$$d_i(w)=\int tr(M^{-1}(w, \theta) A_i(\theta))\, {\rm d}\pi(\theta).$$

\noindent {\bf Algorithm~I} 
\begin{description}
\item
Set $w^{(0)}=(w_1^{(0)}, \ldots, w_n^{(0)})\in \Omega$.  That is, $w_i^{(0)}>0$ for all $i$.
\item
For $t=0,1,\ldots$, compute 
\begin{equation}
\label{alg1}
w_i^{(t+1)}=w_i^{(t)} \frac{d_i\left(w^{(t)}\right)-\alpha^{(t)}}{m-\alpha^{(t)}},\quad i=1, \ldots, n, 
\end{equation}
where $\alpha^{(t)}$ satisfies 
\begin{equation}
\label{bd}
\alpha^{(t)}\leq \frac{1}{2}\min_{i=1}^n d_i\left(w^{(t)}\right).
\end{equation}
\item 
Iterate until convergence.
\end{description}

A commonly used convergence criterion is 
\begin{equation}
\label{conv}
\max_{i=1}^n d_i\left(w^{(t)}\right)\leq m+\epsilon,
\end{equation}
where $\epsilon$ is a small positive constant.  This is based on the general equivalence theorem (Kiefer and Wolfowitz 1960; Whittle 1973), which characterizes any maximizer of $\phi(w)$, $\hat{w}$, by $\max_{i=1}^n d_i(\hat{w})= m$. 

Algorithm~I slightly generalizes the one proposed by Dette et al. (2008).  In a regression context, Dette et al. (2008) prove that Algorithm~I is monotonic for D-optimality, i.e., when $\pi(\theta)$ is a point mass.  Numerical examples support the conjecture that Algorithm~I is monotonic for Bayesian D-optimality in general.  We shall confirm this conjecture (Theorem~\ref{thm1}). 

\begin{theorem}
\label{thm1}
Assume $\phi(w)$ is finite for at least one $w\in \Omega$.  Let $w^{(t)}, w^{(t+1)}\in \Omega$ satisfy (\ref{alg1}) and (\ref{bd}).  Then we have  $$\phi\left(w^{(t+1)}\right)\geq \phi\left(w^{(t)}\right),$$ 
with equality only if $w^{(t+1)}=w^{(t)}$. 
\end{theorem} 

Once strictly monotonicity is established, global convergence holds under mild conditions.  We state such a result where $\alpha^{(t)}$ takes a convenient parametric form. 
\begin{theorem}
\label{thm2}
Assume $\phi(w)$ is finite for at least one $w\in \Omega$.  Let $w^{(t)}$ be a sequence generated by (\ref{alg1}), starting with $w^{(0)}\in \Omega$.  Assume 
\begin{equation}
\label{alpha}
\alpha^{(t)}=\frac{a}{2} \min_{i=1}^n d_i\left(w^{(t)}\right),
\end{equation}
where $a\in [0, 1]$ is a constant.  Then all limit points of $w^{(t)}$ are global maxima of $\phi(w)$ on $\bar{\Omega}$. 
\end{theorem}

Note that a limit point of $w^{(t)}$ may have some zero coordinates, although we require the starting value $w_i^{(0)}>0$ for all $i$.  Also, $\alpha^{(t)}$ changes from iteration to iteration.  Nevertheless, based on Theorem~\ref{thm1}, Theorem~\ref{thm2} can be established by an argument similar to that of Theorem~2 of Yu (2010a) (details omitted). 

A natural question is the choice of $\alpha^{(t)}$.  Given similar computing costs per iteration, it is reasonable to choose $\alpha^{(t)}$ based on the convergence rate.  For D-optimal designs, i.e., when $\pi(\theta)$ is a point mass, Yu (2010b) analyzes the convergence rate of Algorithm~I.  We expect the results to carry over to Bayesian D-optimality.  Specifically, treating the $\alpha^{(t)}\equiv 0$ case as the basic algorithm, we can view Algorithm~I with $\alpha^{(t)}>0$ as an overrelaxed version (in the sense of successive overrelaxation in numerical analysis; see Young 1971).  At each iteration, overrelaxation multiplies the step length of the basic algorithm by $m/(m-\alpha^{(t)})$.  Noting $\sum_{i=1}^n d_i(w) w_i=m$ and (\ref{bd}), we have $\alpha^{(t)}\leq m/2$.  That is, $m/(m-\alpha^{(t)})\leq 2$.  Thus, roughly speaking, overrelaxation can at most double the speed of the basic algorithm.  A caveat is that, when the basic algorithm is very fast, overrelaxation may slow it down due to overshooting.  Nevertheless, the examples provided by Dette et al.\ (2008) indicate that this rarely happens in practice.  The slowness of the basic algorithm is usually the main concern. 

For the rest of this section we illustrate our theoretical results with a logistic regression model 
$$\Pr(y=1|x, \theta)=1-\Pr(y=0|x,\theta)= \left(1+\exp\left(-x^\top \theta\right)\right)^{-1}.$$
More examples can be found in Dette et al.\ (2008).  Consider the design space 
$$\mathcal{X}_1=\left\{x_i=(1, i/10-1)^\top:\ i=1,\ldots, 30\right\}.$$
The Fisher information for $\theta$ from a unit assigned to $x_i$ is 
$$A_i(\theta)=x_i x_i^\top\frac{\exp(\eta_i)}{(1+\exp(\eta_i))^2},\quad \eta_i\equiv x_i^\top \theta.$$
Suppose the distribution $\pi(\theta)$ assigns probability $1/25$ to each point in the following set 
$$\left\{(i, j)^\top:\ i, j=-2, -1, 0, 1, 2\right\}.$$
We implement Algorithm~I to compute the Bayesian D-optimal design.  The $\alpha^{(t)}$ is specified by (\ref{alpha}) with several choices of $a$.  Each algorithm is started at the uniform design $w^{(0)}=(1/30, \ldots, 1/30)$, and we consider two convergence criteria corresponding to (\ref{conv}) with $\epsilon=10^{-3}$ and $\epsilon=10^{-4}$ respectively.  Table~1, which records the iteration counts, shows the advantage of using larger $a$ ($a\leq 1$).  The large iteration counts when $\epsilon=10^{-4}$ illustrate the potential slow convergence of Algorithm~I.  We also display the optimality criterion $\phi\left(w^{(t)}\right)$ in Figure~1.  As Theorem \ref{thm1} claims, $\phi\left(w^{(t)}\right)$ increases monotonically for each algorithm. 

\begin{table}
\caption{Iteration counts for Algorithm~I with $\alpha^{(t)}$ specified by (\ref{alpha}).}
\begin{center}
\begin{tabular}{r|rrrrr}
\hline
                 & $a=0$  & $a=1/4$ & $a=1/2$  & $a=3/4$  & $a=1$ \\
\hline
$\epsilon=10^{-3}$ & 929  & 823     & 718      & 613      & 507   \\
$\epsilon=10^{-4}$ & 4112 & 3643    & 3175     & 2706     & 2238  \\
\hline
\end{tabular}
\end{center}
\end{table}

\begin{figure}
\begin{center}
\psfrag{phi(w)}{$\phi(w)$}
\includegraphics[width=3.8in, height=5.5in, angle=270]{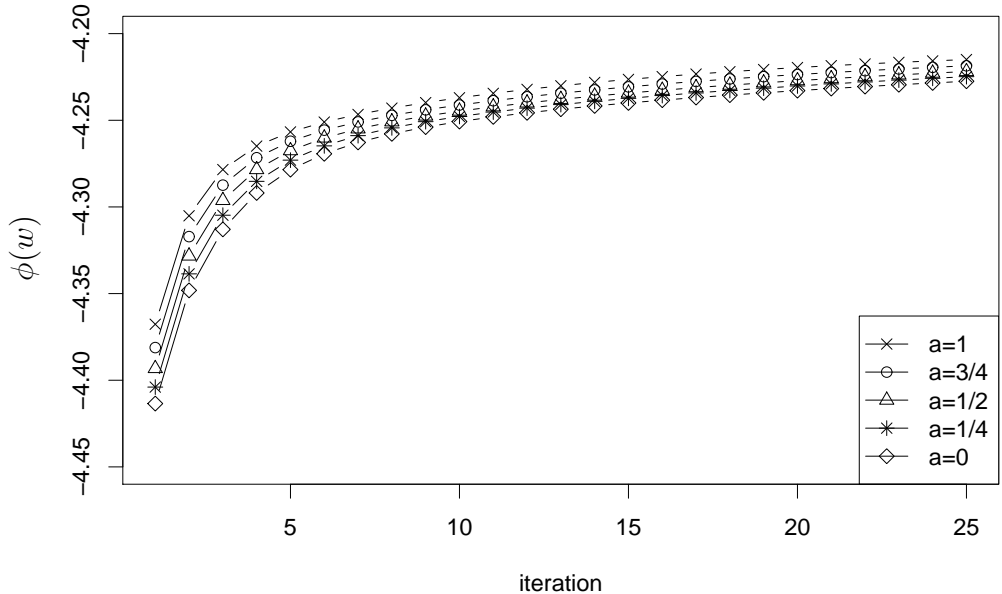}
\end{center}
\caption{Monotonicity of $\phi\left(w^{(t)}\right)$ for Algorithm~I.}
\end{figure}

Table~2 records the design weights as calculated by Algorithm~I with $a=1$.  Note that, as the more stringent  criterion $\epsilon=10^{-4}$ is adopted, the weights assigned to the middle cluster $x_i,\ i=14,\ldots, 18,$ become more concentrated around $x_{16}$.  One interpretation is that Algorithm~I sometimes has difficulty apportioning mass among adjacent design points, and therefore the convergence is slow.  This also hints at potential remedies for the slow convergence.  For computing D-optimal designs, Yu (2010c) proposes a ``cocktail algorithm'' that combines three different strategies for fast monotonic convergence.  One of the ingredients, a special case of Algorithm~I, is a multiplicative algorithm (Silvey et al.\ 1978).  Another ingredient is a strategy that facilitates mass exchange between adjacent design points.  There is no conceptual problem extending the cocktail algorithm to Bayesian D-optimality.  We are investigating such extensions and will report the findings in future works. 

\begin{table}
\caption{Output (design weights) of Algorithm~I with $a=1$ and two convergence criteria.}
\begin{center}
\begin{tabular}{r|rrrrrrr}
\hline  design points & $x_1$  & $x_{14}$ & $x_{15}$  & $x_{16}$ & $x_{17}$ & $x_{18}$ & $x_{30}$ \\
\hline
       output ($\epsilon=10^{-3}$) & 0.434 & 0.006 & 0.073   & 0.114  & 0.035 & 0.003  & 0.334    \\
       output ($\epsilon=10^{-4}$) & 0.435 & 0.000 & 0.026   & 0.204  & 0.002 & 0.000  & 0.334    \\
\hline 
\end{tabular}
\end{center}
\end{table}

\section{Monotonicity of Algorithm~I}
This section proves Theorem~\ref{thm1}.  We need Lemma~\ref{lem1}, which slightly extends Lemma 1 of Dette et al.\ (2008). 

\begin{lemma}
\label{lem1}
For fixed $\theta$, $\det M(w, \theta)$ is a polynomial in $w_1, \ldots, w_n$ with nonnegative coefficients. 
\end{lemma}

\begin{proof}
Let $I_m$ denote the $m\times m$ identity matrix, and define an $m\times (mn)$ matrix $G$ by
$$G=(G_1, \ldots, G_{mn})\equiv \left(A_1^{1/2}(\theta), \ldots, A_n^{1/2}(\theta)\right).$$
We have 
$$M(w, \theta)=G ({\rm Diag}(w)\otimes I_m) G^\top.$$
The Cauchy-Binet formula (Horn and Johnson 1990, Chapter 0) yields 
$$\det M(w, \theta)=\sum_{1\leq i_1<\cdots<i_m\leq mn} h(i_1, \ldots, i_m) u_{i_1} \cdots u_{i_m},$$
where $h(i_1, \ldots, i_m)=\det^2 (G_{i_1}, \ldots, G_{i_m}),$ and $u_i$ denotes the $i$th diagonal of ${\rm Diag}(w)\otimes I_m$.  The claim holds because $u_i$ is equal to one of $w_j$, and $h(i_1, \ldots, i_m)\geq 0$. 
\end{proof}

Lemma \ref{lem2} serves as a building block for constructing our minorization-maximization strategy. 
\begin{lemma}
\label{lem2}
Let $g(w)$ be a nonzero polynomial in $w=(w_1, \ldots, w_n)$ with nonnegative coefficients.  Define 
$$Q(w; \tilde{w})=\sum_{i=1}^n \frac{\partial{g(\tilde{w})}}{\partial w_i} \frac{\tilde{w}_i}{g(\tilde{w})} \log w_i,\quad w, \tilde{w}\in \Omega.$$
Then we have 
$$Q(w; \tilde{w})- Q(\tilde{w}; \tilde{w})\leq \log g(w)-\log g(\tilde{w}).$$ 
\end{lemma}

\begin{proof}
Write $g(w)=\sum_{j=1}^J c_j f_j(w)$ where $c_j> 0$ and $f_j(w)$ are monomials in $w$.  We have
$$\sum_{i=1}^n \frac{\partial f_j(\tilde{w})}{\partial w_i} \tilde{w}_i \log w_i = f_j(\tilde{w}) \log f_j(w),\quad j=1, \ldots, J,$$
because $f_j$ are monomials.  Multiplying both sides by $c_j/g(\tilde{w})$ and then summing over $j$ yield 
$$Q(w; \tilde{w}) =\sum_{j=1}^J \frac{c_j f_j(\tilde{w})}{g(\tilde{w})} \log f_j(w).$$
Hence 
\begin{align*}
Q(w; \tilde{w}) -Q(\tilde{w}, \tilde{w}) -\log g(w)+\log g(\tilde{w}) &= \sum_{j=1}^J \frac{c_j f_j(\tilde{w})}{g(\tilde{w})} \log \frac{c_j f_j(w)/g(w)}{c_j f_j(\tilde{w})/g(\tilde{w})}\\
&\leq \log\left(\sum_{j=1}^J \frac{c_j f_j(w)}{g(w)}\right)\\
&=0,
\end{align*}
where the inequality holds by Jensen's inequality applied to the concave function $\log x$. 
\end{proof}
Lemma~\ref{lem3} is implicit in Dette et al.\ (2008).  The proof is included for completeness. 
\begin{lemma}
\label{lem3}
Define $Q(w)=\sum_{i=1}^n q_i \log w_i,\ q, w\in \Omega$.  For a fixed $w$, let $\alpha$ be a scalar that satisfies 
$$\alpha\leq \frac{1}{2}\min_{i=1}^n \frac{q_i}{w_i}.$$ 
Then we have 
$$Q(\hat{w}) \geq Q(w),\quad \hat{w} \equiv \frac{q - \alpha w}{1-\alpha},$$
with equality only if $\hat{w}=w$. 
\end{lemma}
\begin{proof}
Letting $r_i=q_i/w_i$, we have 
\begin{align*}
Q(\hat{w})-Q(w) &=\sum_{i=1}^n w_i r_i \log \frac{r_i -\alpha}{1-\alpha}\\
  &\geq \bar{r} \log \frac{\bar{r}-\alpha}{1-\alpha}\\
  &=0,
\end{align*}
where $\bar{r}=\sum_i w_i r_i=1$ (hence the last equality), and the inequality follows by Jensen's inequality applied to the function $x\log(x-\alpha)$, which is convex on $x\geq \max\{0,\, 2\alpha\}$.  Hence $Q(\hat{w})\geq Q(w)$.  By strict convexity, equality holds only when $r_i=\bar{r}=1$ for all $i$, i.e., when $\hat{w}=w$. 
\end{proof}

\begin{proof}[Proof of Theorem \ref{thm1}]
It is easy to see that, if $\phi(w)$ is finite for any $w\in \Omega$, then it is finite for all $w\in \Omega$.  Define $g(w, \theta)=\det M(w, \theta)$.  Because $\phi(w)$ is finite, we have $g(w,\theta)>0$ almost surely with respect to $\pi(\theta)$.  By Lemma \ref{lem1}, for fixed $\theta$, $g(w, \theta)$ is a polynomial in $w$ with nonnegative coefficients.  Define ($w, \tilde{w}\in \Omega$) 
\begin{align*}
Q(w; \tilde{w}|\theta) &\equiv \sum_{i=1}^n \frac{\partial{g(\tilde{w}, \theta)}}{\partial w_i} \frac{\tilde{w}_i}{g(\tilde{w}, \theta)} \log w_i\\
 &=\sum_{i=1}^n tr(M^{-1}(\tilde{w}, \theta) A_i(\theta)) \tilde{w}_i\log w_i. 
\end{align*}
By Lemma \ref{lem2}, we have  
$$Q(w; \tilde{w}|\theta)- Q(\tilde{w}; \tilde{w}|\theta) \leq \log g(w, \theta) -\log g(\tilde{w}, \theta).$$
Integration yields 
\begin{align*}
\sum_{i=1}^n d_i(\tilde{w}) \tilde{w}_i \log \frac{w_i}{\tilde{w}_i} &=\int \left[Q(w; \tilde{w}|\theta)-Q(\tilde{w}; \tilde{w}|\theta)\right]\, {\rm d}\pi(\theta)\\
&\leq \int \left[\log g(w, \theta) -\log g(\tilde{w}, \theta)\right]\, {\rm d}\pi(\theta)\\
&= \phi(w) -\phi(\tilde{w}). 
\end{align*}
That is, the function 
$$Q(w; \tilde{w})=\sum_{i=1}^n d_i(\tilde{w}) \tilde{w}_i \log \frac{w_i}{\tilde{w}_i}+\phi(\tilde{w})$$
satisfies $Q(w; \tilde{w})\leq \phi(w)$ and $Q(w; w)=\phi(w)$.  This forms the basis of minorization-maximization. 

Suppose $w^{(t)}, w^{(t+1)}\in \Omega$ are related by (\ref{alg1}).  Applying Lemma~\ref{lem3} with $q_i=d_i\left(w^{(t)}\right) w^{(t)}_i/m$ (note that $\sum_{i=1}^n q_i=1$), we get 
\begin{equation}
\label{qinc}
Q\left(w^{(t+1)}; w^{(t)}\right)\geq Q\left(w^{(t)}; w^{(t)}\right),
\end{equation}
as long as (\ref{bd}) holds.  Thus $\phi\left(w^{(t+1)}\right)\geq Q\left(w^{(t+1)}; w^{(t)}\right)\geq Q\left(w^{(t)};  w^{(t)}\right)=\phi\left(w^{(t)}\right),$
and monotonicity is proved.  Lemma~\ref{lem3} implies that equality holds in (\ref{qinc}) only when $w^{(t+1)}=w^{(t)}$.  Hence the monotonicity is strict. 
\end{proof}

{\bf Remark 1.} Theorem~\ref{thm1} assumes that $w^{(t)}, w^{(t+1)}\in \Omega,$ i.e., they have all positive coordinates.  This assumption can be relaxed.  Inspection of the above proof shows that strict monotonicity holds as long as $w^{(t)}\in \bar{\Omega}$ and $\phi\left(w^{(t)}\right)$ is finite. 

{\bf Remark 2.}  The arguments of Yu (2010a), based on two layers of auxiliary variables, can be extended to prove the monotonicity of (\ref{alg1}) assuming $\alpha^{(t)}\equiv 0$.  This is however weaker than Theorem~\ref{thm1} in the present form. 

\section*{Acknowledgments}
This work is partly supported by a CORCL special research grant from the University of California, Irvine.  The author would like to thank Don Rubin, Xiao-Li Meng, and David van Dyk for introducing him to the field of statistical computing.

\end{document}